\documentclass[conference]{IEEEtran}

\usepackage{epsfig}
\usepackage{cite}
\usepackage{latexsym,amsmath,amsthm}
\usepackage{graphicx}
\usepackage{epic}
\usepackage{times}
\usepackage{latexsym}
\usepackage{amssymb}
\usepackage{amsfonts}
\usepackage{psfrag}

\newtheorem{thm}{Theorem}[section]

\newtheorem{lem}[thm]{Lemma}
\newtheorem{definition}[thm]{Definition}

\newcommand{\ceil}[1]{\left\lceil #1 \right\rceil}
\newcommand{\dl}{{\tt l}}
\newcommand{\dr}{{\tt r}}
\newcommand{\brc}[1]{\left(#1\right)}
\newcommand{\ul}[1]{\underline{#1}}

\begin{document}

% CAPITAL LETTERS IN TITLE!
\title{\Large{Rate-Equivocation Optimal Spatially Coupled LDPC Codes for the BEC Wiretap Channel}}
\author{
\IEEEauthorblockN{
Vishwambhar Rathi\IEEEauthorrefmark{1}, 
R\"udiger Urbanke\IEEEauthorrefmark{2}, 
Mattias Andersson\IEEEauthorrefmark{1},  and 
Mikael Skoglund\IEEEauthorrefmark{1}}

\IEEEauthorblockA{\IEEEauthorrefmark{1}School of Electrical
Engineering and the ACCESS Linnaeus Center, \\Royal Institute of
Technology (KTH),\\ Stockholm, Sweden\\Email: \{vish, amattias,
skoglund@ee\}.kth.se}

\IEEEauthorblockA{\IEEEauthorrefmark{2}School of Computer and Communication Sciences\\ 
EPFL, Lausanne, Switzerland
\\Email: ruediger.urbanke@epfl.ch}
}
\maketitle

\begin{abstract}
We consider transmission over a wiretap channel where both the main
channel and the wiretapper's channel are Binary Erasure Channels
(BEC). We use convolutional LDPC ensembles based on the coset encoding
scheme. More precisely, we consider regular two edge type
convolutional LDPC ensembles. We show that such a construction
achieves the whole rate-equivocation region of the BEC wiretap
channel.

Convolutional LDPC ensemble were introduced by Felstr\"om and
Zigangirov and are known to have excellent thresholds.  Recently,
Kudekar, Richardson, and Urbanke proved that the phenomenon of
``Spatial Coupling'' converts MAP threshold into BP threshold for
transmission over the BEC.

The phenomenon of spatial coupling has been observed to hold for 
general binary memoryless symmetric channels.  Hence, we conjecture that our construction
is a universal rate-equivocation achieving construction when the
main channel and wiretapper's channel are binary memoryless symmetric
channels, and the wiretapper's channel is degraded with
respect to the main channel.
\end{abstract}
%\begin{keywords}
% up to 5 keywords
% fill in: Source coding, quantization, channel coding, correlation.
% Physical Layer Security, Nested Codes
%\end{keywords}

\section{Introduction}
The wiretap channel was introduced by Wyner in \cite{Wyn75}. The basic
diagram is depicted in Figure~\ref{fig:channel}. We consider the
setting when both channels are Binary Erasure Channels (BEC). We
denote a BEC with erasure probability $\epsilon$ by BEC($\epsilon$).
In a wiretap channel, Alice is communicating a message $W$ to Bob. The
message is uniformly chosen from the message set $\mathcal{W}_n$ and it
is sent through the main channel, which is a BEC($\epsilon_m$).  Alice
encodes $W$ as an $n$ bit vector $\ul{X}$ and transmits it. Bob
receives a partially erased version of $\ul{X}$, denote it by
$\ul{Y}$.  Eve is observing $\ul{X}$ via the wiretapper's channel,
which is a BEC($\epsilon_w$).  Let $\ul{Z}$ denote the observation of
Eve.  We denote this wiretap channel by BEC-WT($\epsilon_m,
\epsilon_w$). In order to fulfill the requirement of degradation of
the wiretapper's channel w.r.t. the main channel, we assume that
$\epsilon_w \geq \epsilon_m$.  We denote the capacity of the main
channel and wiretapper's channel by $C_m = 1-\epsilon_m$ and $C_w =
1-\epsilon_w$, respectively.  \psfrag{Alice}{Alice} \psfrag{Bob}{Bob}
\psfrag{Eve}{Eve} \psfrag{W}{$W$} \psfrag{X}{$\ul{X}$}
\psfrag{Y}{$\ul{Y}$} \psfrag{Z}{$\ul{Z}$}
\psfrag{BECm}{BEC$(\epsilon_m)$} \psfrag{BECw}{BEC$(\epsilon_w)$}
\begin{figure}[htbp]
  \centering
  \includegraphics[width=\columnwidth]{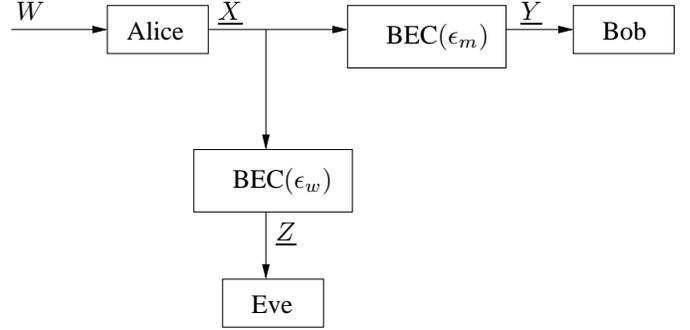}
  \caption{Wiretap channel.}
  \label{fig:channel}
\end{figure}
The encoding of the message $W$ by Alice should be such that Bob is
able to decode $W$ reliably and that $\ul{Z}$ provides as little
information to Eve as possible about $W$.

Assume that transmission takes place using the code $G_n$ and let
$\hat{W}$ be the message decoded by Bob.  We define the performance
metric for reliability to be the average error probability
$P_e\brc{G_n}$,
\begin{equation}
	P_e\brc{G_n} = \frac{1}{|\mathcal{W}_n|} \sum_{w \in \mathcal{W}_n} P\brc{\hat{W} \neq w \mid W = w}. 
\end{equation} 
We use the normalized equivocation $R_e$ as the performance metric for
secrecy,
\begin{equation}
  R_e\brc{G_n} = \frac{1}{n} H\brc{W \mid \ul{Z}}. 
\end{equation}
The rate $R$ of the coding scheme for the intended receiver Bob is
given by
\begin{equation}
  R(G_n) = \frac{\log_2\brc{|\mathcal{W}_n|}}{n}. 
\end{equation}
We say that a rate-equivocation pair $(R, R_e)$ is achievable using a
sequence of codes $G_n$ if
\begin{equation}\label{eq:achcrit}
 \lim_{n \to \infty} R(G_n) = R, \ \lim_{n \to \infty} P_e\brc{G_n} = 0, \  R_e \leq \liminf_{n \to \infty} R_e(G_n). 
\end{equation}
%\begin{equation}
%\frac{\log_2\brc{|\mathcal{W}|}}{n} > R-\epsilon.
%\end{equation}
The achievable rate-equivocation pair $(R, R_e)$ for the
BEC-WT($\epsilon_m, \epsilon_w$) is given by \cite{CsK78},
\begin{equation}
R_e \leq R \leq C_m, \quad 0 \leq R_e \leq C_m-C_w. 
\end{equation}
\begin{figure}[htp]
  \centering \setlength{\unitlength}{1bp}%
  \begin{picture}(113,90)
    \put(0,0){\includegraphics[scale=0.9]{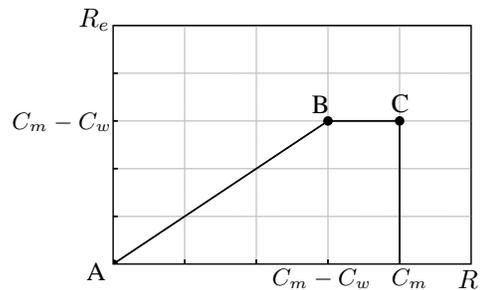}} {
      % \multiputlist(28,-9)(26,0)[b]{$0.2$,$0.4$,$0.6$,$0.8$,$1.0$}
      \put(-13, 88){\makebox(0,0)[lb]{$R_e$}}
      \put(-10, -6){\makebox(0,0)[lb]{A}}
      \put(130,-9){\makebox(0,0)[lb]{$R$}}
      \put(105,-9){\makebox(0,0)[lb]{\small $C_m$}}
      \put(60,-9){\makebox(0,0)[lb]{\small $C_m-C_w$}}
      \put(-38, 50){\makebox(0,0)[lb]{\small $C_m-C_w$}}
      \put(75, 57){\makebox(0,0)[lb]{B}}
      \put(105, 57){\makebox(0,0)[lb]{C}}
      % \put(48,20){\makebox(20,20)[lb]{b}}
      % \multiputlist(-15,17)(0,18)[l]{$0.2$,$0.4$,$0.6$,$0.8$,$1.0$}
}
\end{picture}
\caption{\label{fig:rate_equ} Achievable rate equivocation region for
  BEC-WT($\epsilon_m, \epsilon_w$).}
\end{figure}

Note that we consider weak notion of secrecy as opposed to the strong notion \cite{MaW00, LPS09}. 
%secrecy would impose the following achievability constraint on
%equivocation as opposed to the ``normalized'' constraint in {\bf This is not well defined}
%(\ref{eq:achcrit}),
%\[
%\liminf_{n \to \infty} H\brc{W \mid \ul{Z}} - nR_e \geq 0
%\]

From Figure~\ref{fig:rate_equ}, we see that the boundary of the
achievable rate-equivocation region is composed of two branches,
namely AB and BC. The branch AB corresponds to achieving \emph{perfect
  secrecy}, i.e., $R_e=R \leq C_m-C_w$. The point B corresponds to the
\emph{secrecy capacity}, the highest rate at which perfect secrecy is
possible. The branch BC corresponds to achieving information rates
higher than secrecy capacity.  However, in this case some information
``leaks'' to Eve (the equivocation in this case is strictly smaller
than the rate).

Recently, it has been shown that, using Arikan's polar codes
\cite{Ari09}, it is possible to achieve the whole rate-equivocation
region \cite{ARTKS10,HoS10,MaV10,OzG10}. In
this paper, we show that convolutional LDPC codes achieve the whole
rate-equivocation region for the BEC wiretap channel. Why might this
be of interest?  Compared to polar codes, convolutional LDPC ensembles
have two potential advantages. First, these codes are not only
asymptotically very good but they are know to be competitive with the
best known codes already for modest lengths. Second, convolutional
LDPC ensembles have the potential of being {\em universal}, i.e., one
and the same code is optimal for a large class of channels. Before
discussing this point in more detail, let us first quickly review the
literature on convolutional LDPC codes.

Convolutional LDPC codes were introduced by Felstr\"om and Zigangirov
and were shown to have excellent thresholds \cite{FeZ99}. There has
been a significant amount of work done on convolutional-like LDPC
ensembles \cite{EnZ99,LTZ01,TSSFC04,SLCZ04,LSZC05,LFZC09}, and
see in particular the literature review in \cite{KRU10}. The
explanation for the excellent performance of convolutional-like or
``spatially coupled'' codes over the BEC was given by Kudekar,
Richardson, and Urbanke in \cite{KRU10}.  (In the following, we also
use the term spatially coupled codes when we refer to convolutional
like codes.) More precisely, it was shown in \cite{KRU10} that the
phenomenon of spatial coupling has the effect of converting MAP
threshold of underlying ensemble to BP threshold for BEC and regular
LDPC codes. This phenomenon has been observed to hold in general over
Binary Memoryless Symmetric (BMS) channels, see \cite{KMRU10,LMFC10}.

Thus, when point-to-point transmission is considered over BMS
channels, regular convolutional-like LDPC ensembles are conjectured to
be {\it universally} capacity achieving. This is because the MAP
threshold of regular LDPC ensembles converges to the Shannon threshold
for BMS channels as their left and right degrees are increased by
keeping the rate fixed. To date there is only empirical evidence for
this conjecture. But should in the future a proof be found that
spatially coupled codes are indeed universal for point-to-point
channels, then this would immediately imply that our construction for
the wiretap channel is also universal.

Let us summarize. Our two main motivations for considering code constructions
for the wire-tap channel based on spatially coupled codes is that
these codes perform very well already for modest code lengths and that they
have the potential to be universal.

In \cite{TDCMM07} and \cite{Poor2} coset encoding scheme based sparse
graph codes were given. It was shown in \cite{RATKS10} that
a two edge type LDPC code is a natural candidate for the coset encoding scheme 
 and optimized degree distributions were presented. In the
next section we describe our code design method using spatially coupled
codes. %It is based on the code construction method of \cite{TDCMM07} .

%\section{System model}
%This problem was previously considered in \cite{Thangaraj} and their results
%were restated using the terminology of nested codes \cite{Zamir} in
%\cite{Poor2} as:
%
%\begin{lem}
%  Consider a nested LDPC code sequence $(\mathcal C_0(n),\mathcal C_1(n))$, where the fine code $\mathcal C_0(n)$ of rate $R_0$ has an erasure rate threshold $\delta_0$ and the coarse code $\mathcal C_1 (n)$ of rate $R_1$ is a capacity achieving LDPC code sequence for the binary erasure channel. Suppose that the secure nested code sequence $(\mathcal C_0(n),\mathcal C_1(n))$ is transmitted over a BEC-WT$(\epsilon_0,\epsilon_1)$. If
%  \begin{equation}
%    \delta_0 \geq \epsilon_0 \ \mathrm{and} \ R_1 \leq 1 - \epsilon_1,
%  \end{equation}
%then the code sequence can be successfully transmitted over the erasure wiretap channel with perfect secrecy.
%\end{lem}
%Here BEC-WT$(\epsilon_0,\epsilon_1)$ denotes a wire tap channel where the main channel is a BEC$(\epsilon_0)$ and the wiretapper's channel is a BEC$(\epsilon_1)$. No explicit construction of code sequences is given in \cite{Thangaraj} or \cite{Poor2}.

\section{Code Construction}
We first describe the coset encoding scheme. 
Let $H$ be an $(1-r) n \times n$ LDPC matrix and let $H_1$
and $H_2$ be the submatrices of $H$ such that
\begin{equation}\label{eq:twoedgemat}
  H = \begin{bmatrix}
    H_1 \\ H_2
  \end{bmatrix},
\end{equation}
where $H_1$ is an $(1-r_1) n\times n$ and $H_2$ is an $R n \times n$
matrix.  Let $G^{(1)}_n$ be the code with parity-check matrix $H_1$,
and let $G_n^{(1,2)}$ be the code whose parity-check matrix is
$H$. Assume that Alice wants to transmit an $n R$-bit message
$\ul{S}$. To do this she transmits $\ul{X}$, which is a randomly
chosen solution of
\[
\begin{bmatrix}
  H_1 \\ H_2
\end{bmatrix} \ul{X} = [0 \cdots 0 \ul{S}]^T.
\]
As shown in \cite{TDCMM07}, if $H$ is capacity achieving over the wiretapper's channel 
then $\ul{S}$ is perfectly secure from Eve.  Also, if the threshold of the code $G^{(1)}_n$ 
is higher than the main channel erasure probability $\epsilon_m$ then Bob can recover $\ul{S}$ reliably. We call this wiretap code $G_n$.

The code described by the LDPC matrix $H$ given in (\ref{eq:twoedgemat}) is a two edge type LDPC code. 
The two types of edges are the edges connected to check nodes in $H_1$ 
and those connected to check nodes in $H_2$.  An example of a two edge 
type LDPC code is shown in Figure~\ref{fig:ldpc}.
\psfrag{Type 1 checks}{Type 1 checks}
\psfrag{Type 2 checks}{Type 2 checks}
\psfrag{x1l}{$x_1^{(l)}$}
\psfrag{x2l}{$x_2^{(l)}$}
\psfrag{y1l}{$y_1^{(l)}$}
\psfrag{y2l}{$y_2^{(l)}$}
\begin{figure}[htbp]
  \centering
  \includegraphics[width=\columnwidth]{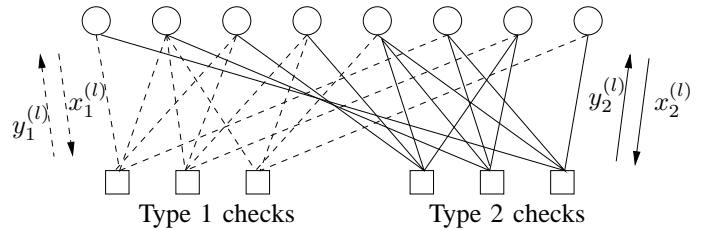}
  \caption{Two edge type LDPC code.}
  \label{fig:ldpc}
\end{figure}

For our purpose it is sufficient to focus on regular two edge type LDPC ensembles.
\begin{definition}[$\{\dl_1, \dl_2, \dr_1, \dr_2\}$ Two Edge Type LDPC Ensemble]
  A $\{\dl_1, \dl_2, \dr_1, \dr_2\}$ two edge type LDPC ensemble of
  blocklength $n$ contains all the bipartite graphs (allowing multiple
  edges between a variable node and a check node) where all the $n$
  variable nodes are connected to $\dl_i$ check nodes of type $i$ and
  all the type $i$ check nodes have degree $\dr_i$, $i \in \{1, 2\}$.
\end{definition}
A protograph of a regular two edge type LDPC code is shown in
Figure~\ref{fig:36protograph}.
\begin{figure}[htp] \begin{centering}
\setlength{\unitlength}{1.0bp}%
\begin{picture}(210,86)(0,0)
\put(0,0)
{
\put(0,0){\rotatebox{0}{\includegraphics[scale=1.0]{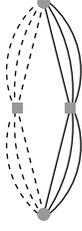}}}
%\put(0,-10){\makebox(0,0)[b]{$\text{-}L$}}
%\put(30,-10){\makebox(0,0)[lb]{$\cdots$}}
%\multiputlist(60,-10)(10,0)[b]{$\text{-}4$,$\text{-}3$,$\text{-}2$,$\text{-}1$,$0$,$1$,$2$, $3$, $4$}
%\put(170,-10){\makebox(0,0)[lb]{$\cdots$}}
%\put(200,-10){\makebox(0,0)[b]{$L$}}
}
\end{picture}
    \caption{A protograph of a two edge type
      LDPC ensemble with $\dl_1=\dl_2=3$ and $\dr_1=\dr_2=6$.
    } \label{fig:36protograph} \end{centering}
\end{figure}

Based on the definition of an $\{\dl, \dr, L, w\}$ ensemble from
\cite{KRU10}, we define the regular spatially coupled two edge
type LDPC ensemble. Before giving this definition, we define 
$\mathcal{T}(\dl)$ to be the set of $w$-tuple of non-negative integers 
which sum to $\dl$. More precisely, 
$\mathcal{T}(\dl)=\{(t_0,\cdots,t_{w-1}):\sum_{j=0}^{w-1} t_j=\dl\}$.\\
{\it Remark:} Note that the $w$-tuple $(t_0,\cdots,t_{w-1})$ is called a 
{\it type} in \cite{KRU10}. We avoid this terminology as we refer to different 
edges in two edge type LDPC ensemble by their type.

\begin{definition}[$\{\dl_1, \dl_2, \dr_1, \dr_2, L, w\}$ Spatially
Coupled Two Edge Type LDPC Ensemble]
Assume that there are $M$ variable nodes at positions $[-L, L]$, $L \in
\mathbb{N}$.  The blocklength of a code in the ensemble is
$n = M(2L+1)$. Every variable node has degree $\dl_1$ with respect to
type $1$ edges and $\dl_2$ with respect to type $2$ edges.  At each
position there are $M$ variable nodes, $\frac{\dl_1}{\dr_1} M$ check
nodes of type $1$ which has degree $\dr_1$, and $\frac{\dl_2}{\dr_2} M$ check nodes of type
$2$ which has degree $\dr_2$. 
  
Assume that for each variable node we order its edges in an arbitrary but fixed 
order. A {\it constellation} $c$ of type $j$ is an $\dl_j$-tuple, 
$c=(c_1,\cdots,c_{\dl_j})$ with elements in $\{0,1,\cdots,w-1\}$, $j \in \{1, 2\}$. 
Its operational significance is that if a variable node at position $i$ has 
type $j$ constellation as $c_j$ then its $k$-th edge of type $j$ is connected to 
a check node at position $i+c_k$, $j \in \{1,2\}$. We denote the set of all the 
type $j$ constellations by $\mathcal{C}_j$. Let $\tau(c)$ be the 
$w$-tuple which counts the occurence of $0, 1, \cdots, w-1$ in $c$. Clearly, if 
$c$ is a type $j$ constellation then $\tau(c) \in \mathcal{T}(\dl_j)$. We impose uniform 
distribution over both the type of constellations. This imposes the following distribution 
over $t \in \mathcal{T}(\dl_j)$ 
\[
	p^{(j)}(t) = \frac{|\{c \in \mathcal{C}_j: \tau(c)=t\}|}{w^{\dl_j}}, \quad j \in \{1, 2\}.
\]
Now we pick $M$ so that $M p^{(1)}(t_1) p^{(2)}(t_2)$ is a natural number for 
$\forall t_1 \in \mathcal{T}(\dl_1), \forall t_2 \in \mathcal{T}(\dl_2)$. For each position 
$i$ pick $M p^{(1)}(t_1) p^{(2)}(t_2)$ which have their type $j$ edges assigned 
according to $t_j$, $j \in \{1, 2\}$. We use a random permutation for each variable and 
type $j$ edge over $\dl_j$ letters to map $t_j$ to a constellation, $j \in \{1, 2\}$. 
Ignoring boundry effects, for each check position $i$, the number of type $j$ edges 
that come from variables at position $i-k$, $k \in \{0,\cdots,w-1\}$, is 
$M \frac{\dl_j}{w}$, $j \in \{1, 2\}$. This implies, it is  exactly a fraction 
$\frac{1}{w}$ of the total number $M \dl_j$ of sockets at position $i$. At the 
check nodes, we distribute this edges by randomly choosing a permutation over 
$M \dl_j$ letters, to the $M \frac{\dl_j}{\dr_j}$ check nodes of type $j$, 
$j \in \{1, 2\}$. 

\end{definition} 
{\it Remark:}
Each of the $\dl_1$ (resp. $\dl_2$) type $1$ (resp. $2$) 
connections of a variable node at position $i$ is uniformly and
independently chosen from the range $[i,\dots,i+w-1]$, where $w$ is
a ``smoothing'' parameter. Similarly, as was remarked in \cite{KRU10}, 
for each check node each edge is roughly independently chosen to be 
connected to one of its nearest $w$ ``left'' neighbors. More precisely, 
the corresponding probability deviates at most by a term of order $1/M$ 
from the uniform distribution. 
%$\dr_1$ (resp. $\dr_2$)
%connections of a type $1$ (resp. $2$) check node at position $i$ is
%independently chosen from the range $[i-w+1,\dots,i]$.

To summarize, a $\{\dl_1, \dl_2, \dr_1, \dr_2, L, w\}$ spatially
coupled two edge type LDPC ensemble is obtained by replacing the
standard regular LDPC ensemble in the $(\dl, \dr, L, w)$ ensemble
(defined in \cite{KRU10}) by a $\{\dl_1, \dl_2, \dr_1, \dr_2\}$ two
edge type LDPC ensemble. The spatial coupling is done such that only the edges of
the same type are coupled together.
\begin{figure}[htp]
\begin{centering}
\setlength{\unitlength}{1bp}%
\setlength{\unitlength}{1bp}%
\begin{picture}(210,100)(-5,-10)
%\put(0,0)
%{
%\put(-5,0){\includegraphics[scale=1.0]{figs/chainL10dlh3}}
%\put(30,-10){\makebox(0,0)[b]{$\text{-}L$}}
%\put(45,-10){\makebox(0,0)[b]{$\cdots$}}
%\multiputlist(60,-10)(10,0)[b]{$\text{-}4$,$\text{-}3$,$\text{-}2$,$\text{-}1$,$0$,$1$,$2$, $3$, $4$}
%\put(155,-10){\makebox(0,0)[b]{$\cdots$}}
%\put(170,-10){\makebox(0,0)[b]{$L$}}
%}
\put(0,0)
{
\put(-5,0){\includegraphics[scale=1.0]{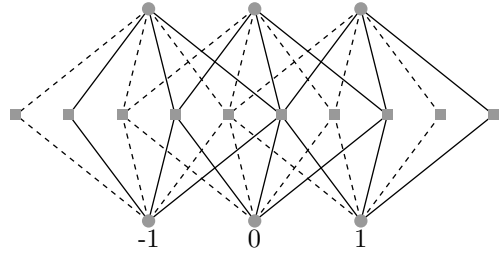}}
\put(55,-10){\makebox(0,0)[b]{$\text{-}1$}}
\put(95,-10){\makebox(0,0)[b]{$0$}}
\put(135,-10){\makebox(0,0)[b]{$1$}}
%\put(30,-10){\makebox(0,0)[lb]{$\cdots$}}
%\multiputlist(60,-10)(10,0)[b]{$\text{-}4$,$\text{-}3$,$\text{-}2$,$\text{-}1$,$0$,$1$,$2$, $3$, $4$}
%\put(150,-10){\makebox(0,0)[lb]{$\cdots$}}
}
\end{picture}

\caption{A coupled chain of protographs of a  two edge type LDPC code with $L=1$ 
for $\dl_1=\dl_2=3$ and $\dr_1=\dr_2=6$. \label{fig:chain}}
\end{centering}
\end{figure}
An example of a protograph of a two edge type LDPC code is shown in 
Figure~\ref{fig:36protograph} and its spatially coupled version is shown 
in Figure~\ref{fig:chain}.

In the next lemma we show that if the degrees of the two types of
check nodes are the same, i.e. if $\dr_1=\dr_2=\dr$, then the $\{\dl_1,
\dl_2, \dr, \dr, L, w\}$ spatially coupled two edge type LDPC ensemble
has the same asymptotic performance as that of the spatially coupled
ensemble $(\dl_1+\dl_2, \dr, L, w)$.
\begin{lem}\label{lem:ensembleeq}
The $\{\dl_1, \dl_2, \dr, \dr, L, w\}$ spatially coupled two edge type
LDPC ensemble has the same BP threshold as the spatially coupled ensemble
$(\dl_1+\dl_2, \dr, L, w)$.
\end{lem}
\begin{proof}
Let $x^{(l,j)}_i$ be the average erasure probability which is emitted by a 
variable node at position $i$ in the $l^{\text{th}}$ iteration along an 
edge of type $j$, $j \in \{1, 2\}$. For $i \notin [-L, L]$, we set 
$x^{(l,j)}_i = 0$. For $i \in [-L, L]$, $j \in \{1, 2\}$, and $l=0$, we set 
$x^{(0,j)}_i = \epsilon$.   

As in \cite{KRU10}, the density evolution recursion for the $\{\dl_1, \dl_2, \dr, \dr, L, w\}$ 
two edge type spatially coupled LDPC ensemble is given by 
\begin{align}
x^{(l, 1)}_i & = \epsilon 
\brc{1-\frac{1}{w}\sum_{p=0}^{w-1} 
\brc{1-\frac{1}{w}\sum_{k=0}^{w-1}x^{(l-1,1)}_{i+p-k}}^{\dr-1}}^{\dl_1-1} \nonumber \\ 
 & \quad \brc{1-\frac{1}{w}\sum_{p=0}^{w-1} 
\brc{1-\frac{1}{w}\sum_{k=0}^{w-1}x^{(l-1,2)}_{i+p-k}}^{\dr-1}}^{\dl_2}, 
\end{align}
\begin{align}
x^{(l, 2)}_i & = \epsilon 
\brc{1-\frac{1}{w}\sum_{p=0}^{w-1} 
\brc{1-\frac{1}{w}\sum_{k=0}^{w-1}x^{(l-1,1)}_{i+p-k}}^{\dr-1}}^{\dl_1} \nonumber \\ 
 &  \quad \brc{1-\frac{1}{w}\sum_{p=0}^{w-1} 
\brc{1-\frac{1}{w}\sum_{k=0}^{w-1}x^{(l-1,2)}_{i+p-k}}^{\dr-1}}^{\dl_2-1}.
%x_1^{(l)} & = \epsilon \brc{y_1^{(l-1)}}^{\dl_1-1} \brc{y_2^{(l-1)}}^{\dl_2}, \\
%x_2^{(l)} & = \epsilon \brc{y_1^{(l-1)}}^{\dl_1} \brc{y_2^{(l-1)}}^{\dl_2-1}.
\end{align}
%When $\dr_1=\dr_2=\dr$, i.e., when check nodes of both types have the
%same degree, then the above two dimensional recursion reduces to the
%following one dimensional recursion
%\begin{equation}
%x^{(l+1)}  = \epsilon \brc{1-(1-x^{(l)})^{\dr-1}}^{\dl_1+\dl_2-1}. 
%\end{equation}
Here $x_i^{(l,1)}=x_i^{(l,2)}$ if $x_i^{(l-1,1)}=x_i^{(l-1,2)}$. Indeed, for
$l=1$ and $i \in [-L, L]$ , $x_i^{(1,1)}=x_i^{(1,2)}=\epsilon$ and for $i
\notin [-L, L]$, $x_i^{(1,1)}=x_i^{(1,2)}=0$. Thus, by induction on number of
iterations $l$, $x_i^{(l,1)}=x_i^{(l,2)}$. Hence we drop the superscript
corresponding to the type of edge and write the density evolution 
recursion as 
\begin{align} 
x_i^{(l)} = \epsilon 
\brc{1-\frac{1}{w}\sum_{p=0}^{w-1} 
\brc{1-\frac{1}{w}\sum_{k=0}^{w-1}x^{(l-1)}_{i+p-k}}^{\dr-1}}^{\dl_1+\dl_2-1}.  
% & \quad \brc{1-\frac{1}{w}\sum_{p=0}^{w-1} 
%\brc{1-\frac{1}{w}\sum_{k=0}^{w-1}x^{(l-1,2)}_{i+p-k}}^{\dr-1}}^{\dl_2} 
\end{align}
This recursion is same as that of $\{\dl_1+\dl_2, \dr, L, w\}$ spatially 
coupled ensemble given in \cite{KRU10}. 
%So, the density evolution
%recursion for the $\{\dl_1, \dl_2, \dr, \dr\}$ two edge type LDPC
%ensemble is the same as that of the standard regular $\{\dl_1+\dl_2,
%\dr\}$ LDPC ensemble. 
This proves the lemma.
\end{proof}

Before proving the main result, we show that regular two edge type
LDPC ensembles $\{\dl_1, \dl_2, \dr, \dr\}$ have the same growth rate
of the average stopping set distribution as that of the standard
regular $\{\dl_1+\dl_2,\dr\}$ LDPC ensemble.
\begin{lem}\label{lem:mindist}
  Consider the $\{\dl_1, \dl_2, \dr, \dr\}$ regular two edge type LDPC
  ensemble with blocklength $n$, $\dl_1 \geq 3$, and positive design
  rate.  Let $N(n, \omega n)$ be the stopping set distribution of a
  randomly chosen code from this ensemble and let $E(N(n, \omega n))$
  be its average.  Then the growth rate of $E(N(n, \omega n))$ is the
  same as that of the standard regular $\{\dl_1+\dl_2, \dr\}$
  ensemble. In particular, the minimum stopping set distance of the
  $\{\dl_1, \dl_2, \dr, \dr\}$ regular two edge type LDPC ensemble
  grows linearly in $n$.
\end{lem}
\begin{proof}
  Using standard counting arguments we obtain
\begin{multline}
  E(N(n, \omega n)) = \\ \binom{n}{n \omega}
  \frac{\text{coef}\brc{p^{(\dr)}(x)^{\frac{\dl_1 n}{\dr}}, x^{\omega
        \dl_1 n}} \text{coef}\brc{p^{(\dr)}(x)^{\frac{\dl_2 n}{\dr}},
      x^{\omega \dl_2 n}}} {\binom{\dl_1 n}{\omega \dl_1 n}
    \binom{\dl_2 n}{\omega \dl_2 n}},
\end{multline}
where $p^{(\dr)}(x) = (1+x)^\dr - \dr x$. Using Stirling's
approximation for binomial terms and the Hayman expansion for the coef
term, see \cite[Appendix D]{RiU08}, we obtain
\begin{multline}\label{eq:stgrowth}
  \lim_{n \to \infty} \frac{\ln\brc{E(N(n, n \omega))}}{n} = (1-\dl_1-\dl_2) h(\omega) \\
  + \frac{\dl_1}{\dr} \ln\brc{p^{(\dr)}(t)} - \omega \dl_1 \ln(t) \\
  + \frac{\dl_2}{\dr} \ln\brc{p^{(\dr)}(t)} - \omega \dl_2 \ln(t),
\end{multline}
where $h(x)\triangleq -x\ln(x)-(1-x)\ln(1-x)$ is the binary entropy
function, all the logarithms are natural logarithms, and $t$ is a
positive solution of
\begin{align}
x \frac{(1+x)^{\dr-1}-1}{(1+x)^{\dr}-\dr x} & = \omega. 
% x_2
% \frac{(1+x_2)^{\dr_2-1}-(1-x_2)^{\dr_2-1}}{(1+x_2)^{\dr_2}+(1-x_2)^{\dr_2}}
% & = \omega.
\end{align}
From (\ref{eq:stgrowth}), we see that the growth rate is the same as
that of the average stopping set distribution of the standard
$\{\dl_1+\dl_2, \dr\}$ regular LDPC ensemble
\cite[Thm. 2]{OVZ05}. Now, the linearity of minimum stopping set
distance immediately follows from \cite[Cor. 7]{OVZ05}.
% \begin{multline}
\end{proof} {\it Remark:} We could have come to this conclusion by
specializing the general result contained in
\cite[Thm. 5]{KADPS09}. But for the convenience of the reader, and
since the above proof is so short, we decided to include a complete
proof.

Lemma~\ref{lem:mindist} and \cite[Lemma 1]{KRU10} imply that $\{\dl_1,
\dl_2, \dr, \dr, L, w\}$ spatially coupled two edge type LDPC
ensembles with variable node degree at least three have a linear
minimum stopping set distance.  This gives us the following lemma on
the block error probability of the $\{\dl_1, \dl_2, \dr, \dr, L, w\}$
ensemble under iterative decoding.

\begin{lem}\label{lem:blockerror}
  Consider transmission over the BEC($\epsilon$) using the $\{\dl_1,
  \dl_2, \dr, \dr, L , w\}$, spatially coupled two edge type LDPC
  ensembles with BP threshold $\epsilon^*$ and blocklength $n$. Let
  $\dl_1 \geq 3$. Assume that $\epsilon < \epsilon^*$.  Denote by
  $P_e^{(B)}$ the block error probability under iterative decoding.
  Then
\[
 \lim_{n \to \infty} n P_e^{(B)} = 0. 
\]
\end{lem}
\begin{proof}
  In fact, a much stronger result is true -- the block error
  probability converges to $0$ exponentially fast. But for our purpose
  we only need that it converges to zero faster than linearly.

  To see why this is correct, fix $\epsilon < \epsilon^*$. Then, for
  any $\delta>0$, there exists an $l$ so that after $l$ iterations of
  DE, the bit error probability is below $\delta/3$. Further, for
  $n=n(l)$, sufficiently large, the expected behavior over all
  instances of the code and the channel deviates from the density
  evolution predictions by at most $\delta/3$.  Finally, by standard
  concentration results (see \cite[Thm. 3.30]{RiU08}) it follows that
  the probability that a particular instance deviates more than
  $\delta/3$ from its average decays exponentially fast in the
  blocklength.

  We summarize, with a probability which converges exponentially fast
  (in the blocklength) to $1$, an individual instance will have
  reached a bit error probability of at most $\delta$ after a fixed
  number of iterations.

  If $\delta$ is chosen sufficiently small, in particular smaller than
  the relative minimum stopping set distance, then we know that the
  decoder can correct the remaining erasures with probability $1$.
\end{proof}

In the following lemma we calculate the design rate of the spatially coupled two edge type ensemble.
\begin{lem}[Design Rate]
The design rate of the spatially coupled two edge type ensemble $(\{\dl_1,\dl_2,\dr_1,\dr_2,L,w\})$ with $w \leq 2L$ is given by
\begin{align}
&R(\dl_1,\dl_2,\dr_1,\dr_2,L,w) =\\
& \brc{1 - \frac{\dl_1}{\dr_1} - \frac{\dl_2}{\dr_2}} - \brc{\frac{\dl_1}{\dr_1} + \frac{\dl_2}{\dr_2}}\frac{w + 1 - 2 \sum_{i=0}^w \brc{\frac{i}{w}}^r}{2L + 1}.
\end{align}
The design rate of the coset encoding scheme for the wiretap channel is given by
\begin{equation}
  R_{\textnormal{des}} = \frac{\dl_2}{\dr_2} - \frac{\dl_2}{\dr_2} \frac{w+1-2\sum_{i=0}^w \brc{\frac{i}{w}}^r}{2L + 1}.
\end{equation}
\end{lem}
\begin{IEEEproof}
  Let $C_1 (C_2)$ be the number of type one (two) check nodes connected to variable nodes and let $V$ be the number of variable nodes. Then $R(\dl_1,\dl_2,\dr_1,\dr_2,L,w) = 1 - C_1/V - C_2/V$ and $R_{\textnormal{des}} = C_2/V$. The calculations then follow from the proof of \cite[Lemma 3]{KRU10}.
\end{IEEEproof}

The number of possible messages $\ul{s}$ of the coset encoding scheme is given by the number of cosets of $G_n^{(1,2)}$ in $G^{(1)}_n$. For a standard LDPC ensemble the design rate is a lower bound on the rates of the codes in the ensemble. This is not true for the coset encoding scheme for the wiretap channel. For example, suppose the rate of $G^{(1)}_n$ equals the design rate, but the rate of $G_n^{(1,2)}$ is higher than its design rate. 
Then there will be fewer cosets than the maximum possible value. This corresponds to the equation
 \[
   \begin{bmatrix}
    H_1 \\ H_2
  \end{bmatrix} \ul{X} = [0 \cdots 0 \ul{S}]^T. 
\]
not having solutions for some $\ul{S}$.

Now, we are ready to state one of our main theorems. It shows that, by
spatial coupling of two edge type LDPC codes, we can achieve perfect
secrecy (the branch AB in Figure~\ref{fig:rate_equ}), and in particular the secrecy capacity (the point B in Figure~\ref{fig:rate_equ}) of the
binary erasure wiretap channel.

\begin{thm}\label{thm:seccap}
  Consider transmission over the BEC-WT($\epsilon_m, \epsilon_w$)
  using spatially coupled regular $\{\dl_1, \dl_2, \dr, \dr, L, w\}$
  two edge type LDPC ensemble.  Assume that the desired rate of
  information transmission from Alice to Bob is $R$, $R \leq
  C_m-C_w$. Let $\dl_1=\ceil{(1-C_w-R) \dr}$ and $\dl_2=\ceil{(1-C_w)
    \dr}-\ceil{(1-C_w-R) \dr}$.  Let $R_e$ be the average (over the
  channel and ensemble) equivocation achieved for the wiretapper.
  Then,
\[
\lim_{\dr \to \infty} \lim_{w \to \infty} \lim_{L \to \infty} \lim_{M \to \infty}\mathbb{E}\brc{P_e(G_n)} = 0,
\]
\[
\lim_{\dr \to \infty} \lim_{w \to \infty} \lim_{L \to \infty} \lim_{M \to \infty} R_e = R.
\]
Let $R(G_n)$ be the rate from Alice to Bob of a randomly chosen code
in the ensemble. Then
\[
\lim_{\dr \to \infty} \lim_{w \to \infty} \lim_{L \to \infty} \lim_{M \to \infty} \textnormal{Pr}(R(G_n) < R) = 0.
\]
\end{thm}
\begin{IEEEproof}
  We first show that the rate from Alice to Bob is $R$ almost
  surely. Let $G_n^{(1,2)}$ be a two edge type spatially coupled code,
  and let $G_n^{(1)}$ be the code induced by its type 1 edges
  only. Then
  \begin{equation}
    R(G_n) = R(G^{(1)}_n) - R(G^{(1,2)}_n).
  \end{equation}
  Since both the two edge type spatially coupled ensemble and the
  ensemble induced by its type 1 edges are capacity achieving we must
  have
\begin{align}
  \lim_{\dr \to \infty} \lim_{w \to \infty} \lim_{L \to \infty} \lim_{M \to \infty} \textnormal{Pr}(R(G_n^{(1)}) > C_w + R) = 0, \\
  \lim_{\dr \to \infty} \lim_{w \to \infty} \lim_{L \to \infty} \lim_{M \to \infty} \textnormal{Pr}(R(G_n^{(1,2)}) > C_w) = 0.
\end{align}
This implies
\begin{equation}
  \lim_{\dr \to \infty} \lim_{w \to \infty} \lim_{L \to \infty} \lim_{M \to \infty} \textnormal{Pr}(R(G_n) < R) = 0.
\end{equation}
The reliability part easily follows from the capacity achieving
property of the spatially coupled ensemble.  This is because the rate
of the ensemble corresponding to type $1$ edges approaches $C_w+R$. As
this ensemble is capacity achieving, its threshold is $1-C_w-R$. As $R
< C_m-C_w$, we see that the threshold is greater than
$\epsilon_m$. This proves reliability.

To bound the equivocation of Eve, using the chain rule we expand the mutual information 
$I(\ul{X}, \ul{S}; \ul{Z})$ in two different ways
\begin{align}
  I(\ul{X}, \ul{S}; \ul{Z}) & = I(\ul{X}; \ul{Z}) + I(\ul{S}; \ul{Z} \mid \ul{X}) \\
  & = I(\ul{S}; \ul{Z}) + I(\ul{X}; \ul{Z} \mid \ul{S}).
\end{align}
As $\ul{S} \to \ul{X} \to \ul{Z}$ is a Markov chain, $I\brc{\ul{S};
  \ul{Z} \mid \ul{X}}=0$. Using $I(\ul{S}; \ul{Z}) =
H(\ul{S})-H(\ul{S} \mid \ul{Z})$, we obtain,
\begin{align}
  \frac{1}{n}H(\ul{S} \mid \ul{Z}) & = \frac{1}{n} \brc{H(\ul{S})+I(\ul{X}; \ul{Z} \mid \ul{S})-I(\ul{X}; \ul{Z})} \\
  & = \frac{1}{n} \brc{H(\ul{S})+H(\ul{X} \mid \ul{S})-H(\ul{X} \mid \ul{Z}, \ul{S})} \nonumber \\
  & \phantom{=}  -\frac{I(\ul{X}; \ul{Z})}{n} \\
  & \geq \frac{1}{n} \brc{H(\ul{X})-H(\ul{X} \mid \ul{Z}, \ul{S})} -
  C_w, \label{eq:equrelation}
\end{align}
where we have used that $H(\ul{S})+H(\ul{X} \mid \ul{S}) =
H(\ul{S},\ul{X}) = H(\ul{X})$ and that $I(\ul{X};\ul{Z})/n \leq C_w$.

Since the ensemble induced by type $1$ edges is capacity
achieving its rate must equal its design rate asymptotically, so
\begin{equation}\label{eq:H(X)}
  \lim_{\dr \to \infty} \lim_{w \to \infty} \lim_{L \to \infty} \lim_{n \to \infty} H(\ul{X})/n = R + C_w.
\end{equation}

Denote the block error probability of decoding $\ul{X}$ from $\ul{Z}$
and $\ul{S}$ by $P_e(\ul{X} \mid \ul{S}, \ul{Z})$. From Fano's
inequality we obtain,
\begin{align}\label{eq:fano}
  \frac{H(\ul{X} \mid \ul{S}, \ul{Z})}{n}\leq \frac{h(P_e(\ul{X} \mid
    \ul{S}, \ul{Z}))}{n} + P_e(\ul{X} \mid \ul{S},
  \ul{Z})(1-\epsilon_w).
\end{align}
Note that, as the two edge type spatially coupled construction is
capacity achieving over the wiretapper's channel, $\lim_{\dr \to
  \infty} \lim_{w \to \infty} \lim_{L \to \infty} \lim_{M \to \infty}
P_e(\ul{X} \mid \ul{S}, \ul{Z}) = 0$.

We now obtain the desired bound on the equivocation by substituting
(\ref{eq:fano}) and (\ref{eq:H(X)}) in (\ref{eq:equrelation}), and
taking the limit $\dr,w,L,M \to \infty$.
\end{IEEEproof}

Note that in the previous theorem our requirement was to have perfect
secrecy.  Hence we constructed spatially coupled two edge type matrix
such that it was capacity achieving over the wiretapper's channel. In
the next theorem we prove that using spatially coupled two edge LDPC
codes, it is possible to achieve an information rate equal to $C_m$,
the capacity of the main channel, and equivocation equal to
$C_m-\epsilon_w$.

\begin{thm}
  Consider transmission over the BEC-WT($\epsilon_m, \epsilon_w$)
  using spatially coupled regular $\{\dl_1, \dl_2, \dr, \dr, L, w\}$
  two edge type LDPC ensemble.  Assume that the desired rate of
  information transmission from Alice to Bob is $R$, $R > C_m-C_w$ and
  $R \leq C_m$. Let $\dl_1=\ceil{(1-C_m) \dr}$ and $\dl_2=\ceil{R
    \dr}$.  Let $R_e$ be the average (over the channel and ensemble)
  equivocation achieved for the wiretapper.  Then,
\[
\lim_{\dr \to \infty} \lim_{w \to \infty} \lim_{L \to \infty} \lim_{M \to \infty}\mathbb{E}\brc{P_e(G_n)} = 0,
\]
\[
\lim_{\dr \to \infty} \lim_{w \to \infty} \lim_{L \to \infty} \lim_{M \to \infty} R_e = C_m - C_w.
\]
Let $R(G_n)$ be the rate from Alice to Bob of a randomly chosen code
in the ensemble. Then
\[
\lim_{\dr \to \infty} \lim_{w \to \infty} \lim_{L \to \infty} \lim_{M \to \infty} \textnormal{Pr}(R(G_n) < R) = 0.
\]
%where the limit $\dl \to \infty$ is taken such that the ratio $\dl/\dr=r$ is 
%kept fixed. 
\end{thm}
\begin{proof}
  The proof that the rate is $R$ asymptotically is the same as in the
  proof of Theorem~\ref{thm:seccap}.

  The reliability part easily follows from the capacity achieving
  property of the spatially coupled ensemble corresponding to type $1$
  edges.  This is because the rate of the ensemble corresponding to
  type $1$ edges approaches $C_m$. As this ensemble is capacity
  achieving, its threshold is $\epsilon_m$. This proves reliability.

  The proof for equivocation is very similar to that of
  Theorem~\ref{thm:seccap}.  From (\ref{eq:equrelation}), we know
  \begin{align}
    \frac{1}{n}H(\ul{S} \mid \ul{Z}) & \geq \frac{1}{n}
    \brc{H(\ul{X})-H(\ul{X} \mid \ul{Z}, \ul{S})} -
    C_w.  \label{eq:equrelationthm2}
  \end{align}
  Since the code induced by type $1$ edges is capacity achieving we
  have
  \begin{equation}\label{eq:thm2t1}
    \lim_{\dr \to \infty} \lim_{w \to \infty} \lim_{L \to \infty} \lim_{n \to \infty} H(\ul{X})/n = C_m.
  \end{equation}
  Note that as the two edge type code has rate $C_m-R$ and is capacity
  achieving, its threshold for the BEC is $1-C_m+R$. As $R>C_m-C_w$,
  the threshold is higher than $\epsilon_w$. As in
  Theorem~\ref{thm:seccap}, given $\ul{S}$ the error probability of
  decoding $\ul{X}$ from $\ul{Z}$, denoted by, $P_e(\ul{X} \mid
  \ul{S}, \ul{Z})$ goes to zero.  Thus (\ref{eq:fano}) holds and we
  obtain %{\bf We had limsup n to infty here before. What should we write?}
\begin{align}\label{eq:thm2t2}
  \lim_{\dr \to \infty} \lim_{w \to \infty} \lim_{L \to \infty} \lim_{M \to \infty} \frac{H(\ul{X} \mid \ul{S}, \ul{Z})}{n} = 0.
\end{align}
We obtain the desired bound on the equivocation by substituting
(\ref{eq:thm2t1}) and (\ref{eq:thm2t2}) in (\ref{eq:equrelationthm2}),
and taking the limit $\dr,w,L,M \to \infty$.
\end{proof}

\section{Numerical Results}
We have rigorously shown the optimality of the $\{\dl_1, \dl_2, \dr_1, \dr_2, L, w\}$ ensemble. 
In this section, we briefly discuss the performance of the 
$\{\dl_1, \dl_2, \dr_1, \dr_2, L\}$ ensemble, which is the two 
edge type extension of the $\{\dl, \dr, L\}$ ensemble discussed in \cite[Sec. II.A]{KRU10}.
Based on the method in \cite{RATKS10}, we numerically evaluate the  
equivocation 
 of  the $\{3, 3, 6, 12, L\}$ ensemble for the BEC-WT($0.5, 0.75$). The results are 
given in Table \ref{tab:eqv_L}.  
We observe that as $L$ increases, the equivocation $R_e$ 
converges to $R$, the rate from Alice to Bob. Thus, the optimality of secrecy performance of the
$\{\dl_1, \dl_2, \dr_1, \dr_2, L\}$ ensemble seems to hold for the wiretap channel. 
The optimality of reliability performance has been conjectured to hold in \cite{KRU10}. 
\begin{table}[ht]
\centering 
\begin{tabular} {|c|c|c|c|c|c|c|}
\hline
$L$ & $20$ & $30$ & $40$ & $50$ & $60$ & $70$\\ 
\hline 
$R$ & $0.2622$  & $0.2582$ & $0.2562$  & $0.255$ & $0.2541$ & $0.2535$\\
\hline
$R_e$ & $0.2276$ & $0.235$ & $0.2387$ & $0.241$ & $0.2425$ & $0.2436$\\
\hline
\end{tabular}
\caption{\label{tab:eqv_L} Rate from Alice to Bob ($R$) and equivocation of Eve ($R_e$) 
for different values of $L$, $M=1000$ for $\{3, 3, 6, 12, L\}$ ensemble.}
\end{table}
\vspace{-0.2in}

\section{Conclusion}
We showed how to achieve the whole rate-equivocation region using spatially 
coupled regular two edge type LDPC codes over the binary erasure wiretap channel. 
As the spatially coupled two edge type 
LDPC codes are conjectured to achieve capacity over general BMS channels, we 
conjecture that our code construction is also universally optimal for 
the class of wiretap channel where the main channel and wiretapper's channel are 
BMS channels and wiretapper's channel is physically degraded with respect to the main 
channel. 

\bibliography{secret}{}
\bibliographystyle{IEEEtran}

\end{document}